%% file: main.tex
\documentclass[11pt,a4paper]{article}

\usepackage[T1]{fontenc}
\usepackage[utf8]{inputenc}
\usepackage[margin=1in]{geometry}
\usepackage{amsmath,mathtools}
\usepackage{amsfonts}
\usepackage{amsthm}
\usepackage{graphicx}
\usepackage{booktabs}
\usepackage{xcolor}
\usepackage{xspace}
\usepackage{tikz}
\usepackage{subcaption}
\usepackage{cite}
\usepackage[hidelinks]{hyperref}
\usepackage{orcidlink}
\usepackage[capitalize,nameinlink,noabbrev]{cleveref}
\usetikzlibrary{patterns}

\newtheorem{theorem}{Theorem}
\newtheorem{lemma}[theorem]{Lemma}
\newtheorem{corollary}[theorem]{Corollary}
\newtheorem{observation}[theorem]{Observation}

\newcommand{\ndn}{\ensuremath{\langle n,d,n \rangle}}
\newcommand{\ddd}{\ensuremath{\langle d,d,d \rangle}}
\newcommand{\dnd}{\ensuremath{\langle d,n,d \rangle}}
\newcommand{\kkk}{\ensuremath{\langle k,k,k \rangle}}
\newcommand{\nnn}{\ensuremath{\langle n,n,n \rangle}}
\newcommand{\anymatmul}[3]{\ensuremath{\langle #1,#2,#3 \rangle}}

\providecommand{\Description}[1]{}

\newcommand{\ndnup}{\ensuremath{ O( \min \{d\sqrt{n}, {n+ d^{2/3}n^{2/3}}\}) }}
\newcommand{\ndnupfields}{\ensuremath{ O( \min \{d\sqrt{n}, {n+ d^{2/3}n^{2/3}}, {n^{2-2/\omega}}\}) }}
\newcommand{\ndnlowcond}{\ensuremath{ \Omega( \frac{d^{14/9}}{n^{2/9}}) }}
\newcommand{\ndnlowuncond}{\ensuremath{ \Omega( \min \{d\sqrt{n}, n\}) }}

\newcommand{\dndup}{\ensuremath{ O( d^{4/3} + \log n) }}
\newcommand{\dndlowcond}{\ensuremath{ \Omega( \frac{d^{2}}{n^{2/3}}) }}
\newcommand{\dndlowuncond}{\ensuremath{ \Omega( \frac{d^{2}}{n}) }}
\newcommand{\dndupfields}{\ensuremath{ O( d^{2-2/\omega} + \log n) }}

\newcommand{\sqn}{\sqrt{n}}
\newcommand{\na}{n^{\alpha}}

\definecolor{myblue}{HTML}{0088cc}
\definecolor{myorange}{HTML}{f26924}
\definecolor{mygrey}{HTML}{4D4D4D}

\title{Rectangular Matrix Multiplication in the Low-Bandwidth Model}

\author{%
	Chetan Gupta\,\orcidlink{0000-0002-0727-160X}\\
	IIT Roorkee, Roorkee, India\\
	\texttt{chetan.gupta@cs.iitr.ac.in}
	\and
	Jukka Suomela\,\orcidlink{0000-0001-6117-8089}\\
	Aalto University, Helsinki, Finland\\
	\texttt{jukka.suomela@aalto.fi}
	\and
	Hossein Vahidi\,\orcidlink{0000-0002-0040-1213}\\
	Aalto University, Helsinki, Finland\\
	\texttt{hossein.vahidi@aalto.fi}%
}
\date{}

\begin{document}
	
\maketitle
	
\begin{abstract}
	We study rectangular matrix multiplication in the low-bandwidth model of distributed computing. There are $n$ computers; initially the input matrices are distributed evenly between computers, and in each communication round every computer can send and receive an $O(\log n)$-bit message. Eventually each computer must output its designated part of the product matrix.
		
	While prior work has focused primarily on square $n \times n$ multiplication under various sparsity assumptions, we study \emph{rectangular} instances with no sparsity assumption.
	We denote by $\langle a,b,c\rangle$ the task of multiplying an $a\times b$ matrix by a $b\times c$ matrix in this model. We concentrate on two natural aspect ratios, $\langle n,d,n\rangle$ and $\langle d,n,d\rangle$, for $d \le n$, and we study how the round complexity depends on $n$ and $d$.
		
	When $d \to n$, both $\langle n,d,n\rangle$ and $\langle d,n,d\rangle$ approach $\langle n,n,n\rangle$, which is the usual task of multiplying square matrices. If we consider multiplication over semirings, the current best upper bound in that case is $O(n^{4/3})$ rounds, and there is a trivial unconditional lower bound of $\Omega(n)$.
		
	We show that for $\langle d,n,d\rangle$, we can achieve the complexity of $\tilde O(d^{4/3})$, which seems like a natural generalization of the upper bound $\tilde O(n^{4/3})$ when $d=n$. However, the case of $\langle n,d,n\rangle$ is fundamentally different, and also exhibits a phase transition. We show that for $d \le \sqrt{n}$, the complexity of $\langle n,d,n\rangle$ is $\Theta(d \sqrt{n})$; we have matching upper and lower bounds. However, the behavior is genuinely different in the region $d \ge \sqrt{n}$, where the upper bound is $O(d^{2/3} n^{2/3})$.
\end{abstract}
	
	
\paragraph*{Acknowledgements.} We are grateful to the anonymous reviewers for their helpful feedback on prior versions of this work.
	
\input{intro}
\input{preliminaries}
\input{ndn}
\input{dnd}

\bibliographystyle{plainurl}
\bibliography{da}
	
\end{document}

%% file: intro.tex
\section{Introduction}\label{sec:intro}

In this work we study the fundamental problem of multiplying rectangular matrices in a massively parallel, bandwidth-limited setting.

\subsection{Setting and Prior Work}

\subparagraph{Centralized setting.}

To set the scene, let us first recall what is known about \emph{square} matrices. Assume we are calculating the matrix product $X = AB$, where $A$ and $B$ are $n \times n$ matrices.

In the classical \emph{centralized} setting, for matrix multiplication over \emph{fields}, we have a long sequence of nontrivial algorithms, starting with Strassen's $O(n^{2.808})$-time algorithm \cite{strassen-1969-gaussian-elimination-is-not-optimal}, and the latest news is the $O(n^{2.371339})$-time algorithm from \cite{alman-duan-etal-2025-more-asymmetry-yields-faster}. These algorithms fundamentally exploit subtraction, and for multiplication over a general semiring not much is known beyond the trivial $O(n^3)$-time algorithms.

\subparagraph{Congested clique.}

The centralized algorithms have direct analogs in the massively parallel setting. Let us first consider the \emph{congested clique} model \cite{lotker-pavlov-etal-2003-mst-construction-in-o-log-log-n}; in this model we have $n$ computers, each computer has a direct communication link to all other computers, computation proceeds in synchronous rounds, and in each round we can send an $O(\log n)$-bit message over each communication link. This is equivalent to a setting where each node can send and receive $O(n \log n)$ bits in total in each round \cite{lenzen-2013-optimal-deterministic-routing-and-sorting}. We do not restrict local computation; we only care about the \emph{number of communication rounds}.

When we study matrix multiplication in the congested clique model, it is natural to assume that computer number $i$ initially holds row (or column) number $i$ of matrices $A$ and $B$, and it has to eventually hold row (or column) number $i$ of the product matrix $X$.

The main result of \cite{censor-hillel-kaski-etal-2019-algebraic-methods-in-the} connects matrix multiplication in the congested clique model with fast centralized matrix multiplication algorithms. Informally, their main result is this: Assume that we have a centralized matrix multiplication algorithm that uses only $O(n^\omega)$ operations, for some $2\le\omega\le 3$. Then there is a matrix multiplication algorithm in the congested clique model that takes only $O(n^{1-2/\omega})$ communication rounds. By plugging in the latest value $\omega \approx 2.371339$ from \cite{alman-duan-etal-2025-more-asymmetry-yields-faster} we obtain an algorithm for fields that runs in $O(n^{0.1566})$ rounds. For semirings the trivial choice of $\omega = 3$ yields an algorithm that runs in $O(n^{1/3})$ rounds.

\subparagraph{Sparse and rectangular matrices.}

However, in many applications we have matrices that are not dense, square matrices. E.g.\ \cite{le-gall-2016-further-algebraic-algorithms-in-the,censor-hillel-leitersdorf-turner-2018-sparse-matrix,censor-hillel-dory-etal-2021-fast-approximate-shortest} have studied matrix multiplication in the congested clique model for sparse and rectangular matrices.

To make our point clear, let us focus on the case that matrix $A$ has dimensions $n \times d$, matrix $B$ has dimensions $d \times n$, and the product matrix $X$ has hence dimensions $n \times n$ for some $d \le n$; we use notation \ndn{} to refer to this task. Furthermore, let us focus on the case of semirings, primarily so that we have convenient round numbers. Recall that when $d = n$, the round complexity in the congested clique model is $O(n^{1/3})$.

Now a particularly relevant prior work is \cite{censor-hillel-leitersdorf-turner-2018-sparse-matrix}. Adapted to this setting, their work shows that, for $d \ge \sqrt{n}$, there is an algorithm with round complexity $O(d^{2/3}/n^{1/3}+1)$ in the congested clique model. In particular, the round complexity of this algorithm smoothly grows from $O(1)$ to $O(n^{1/3})$ as $d$ grows from $\sqrt{n}$ to $n$.

For the congested clique model, this essentially ends the line of research for $d < \sqrt{n}$: if there is a constant-round algorithm, what else could we ask for?

However, this result is mainly an artifact of the way the congested clique model is defined. The model is simply too powerful, as we have $\Theta(n \log n)$ bits of bandwidth per node per round, and it is therefore ill-suited when analyzing massively parallel algorithms when the size of the input is relatively small.

\subparagraph{Low-bandwidth model.}

There is recent work \cite{gupta-hirvonen-etal-2022-sparse-matrix-multiplication,gupta-korhonen-etal-2025-low-bandwidth-matrix} that has studied \emph{sparse} matrix multiplication in the \emph{low-bandwidth} model; this is essentially the same model as what is known as the \emph{node-capacitated clique} or \emph{node-congested clique} in the literature \cite{augustine-ghaffari-etal-2019-distributed-computation-in}. In this model we have simply scaled down the bandwidth by a factor of $n$ in comparison with the congested clique model: in each round, each node can send and receive only one message with $O(\log n)$ bits.

All upper bounds from the congested clique model now directly apply, if we multiply the round complexities by $n$. For example, the square semiring algorithm now has the round complexity $O(n^{4/3})$, and \ndn{} can be solved in $O(d^{2/3}n^{2/3}+n)$ rounds. However, now it is not at all clear what happens when we go below $d = \sqrt{n}$. Intuitively, we would expect to be able to do much better than $\Theta(n)$ rounds for a small $d$, but exactly how does the round complexity scale? For example, would $\Theta(d^{2/3}n^{2/3})$ be the right complexity for all $d$?

\subsection{Our Contributions}

\subparagraph{Problem setting.}
In this work, we study rectangular matrix multiplication in the low-bandwidth model. We consider both of these cases, both for semirings and fields, for $d \le n$:
\begin{itemize}
	\item \ndn{}: multiply an $n \times d$ matrix by a $d \times n$ matrix, and the result is an $n \times n$ matrix.
	\item \dnd{}: multiply a $d \times n$ matrix by an $n \times d$ matrix, and the result is a $d \times d$ matrix.
\end{itemize}
We have $n$ computers. The input is distributed among them in a balanced manner. In each round, each computer can send and receive one $O(\log n)$-bit message (we assume that each matrix element fits in such a message). Eventually each computer must hold its own part of the output.

It turns out that the precise distribution of inputs is not particularly important; our algorithms can be adapted to any balanced input distribution and the lower bounds hold for any balanced input distribution. Hence, conveniently, we can assume that e.g.\ if a matrix has dimensions $n \times d$, computer number $i$ initially holds row $i$, which contains $d$ elements. The output matrix requires a bit more care in the \ndn{} case, since reorganizing the output is expensive, but again we can design our algorithms to work for natural output distributions and our lower bounds hold for any balanced output distribution.

\subparagraph{Results.}

Our results are summarized in \cref{tab:main-results,tab:main-results-fields,fig:bounds}. We present both upper and lower bounds on the round complexity of \ndn{} and \dnd{} in the low-bandwidth model, for semirings and fields. We have both unconditional lower bounds and conditional lower bounds. The conditional lower bounds are conditioned on the current \emph{square} matrix multiplication algorithms being optimal.

Our highlight result is related to \ndn{}. We completely characterize what happens in the region $d \le \sqrt{n}$: in that region we have an upper bound of $O(d\sqrt{n})$ and a near-matching \emph{unconditional} lower bound of $\Omega(d\sqrt{n})$.

Notably, this means that there is a \emph{phase transition} at $d = \Theta(\sqrt{n})$. At this point, the round complexity is $\Theta(n)$. Below this point, round complexity scales in proportion to $\Theta(d)$. However, above this point, the round complexity scales in proportion to $O(d^{2/3})$.

This also highlights how misleading a view of the complexity landscape we would get if we only look at the congested clique model. Only by zooming into the low-bandwidth model can we see that the true complexity of this problem is not, for example, $\Theta(d^{2/3}n^{2/3})$, but we can do strictly better for $d \le \sqrt{n}$.

Our results for \dnd{} show, in particular, that the complexity landscape of this problem is fundamentally different from \ndn{}. In particular, there is no sign of a phase transition, and we can solve \dnd{} faster than \ndn{} across all values of $d < n$. Here it is good to note that the total size of the input matrices is the same, and only the output matrices differ in size; the smaller output matrix significantly helps.

\begin{table*}[p]
	\centering
	\caption{Rectangular matrix multiplication over semirings in the low-bandwidth model.}
	\label{tab:main-results}
		\resizebox{\textwidth}{!}{%
	\begin{tabular}{lllll}
		\toprule
		& \multicolumn{2}{c}{\ndn{}} & \multicolumn{2}{c}{\dnd{}} \\
		\cmidrule(lr){2-3} \cmidrule(lr){4-5}
		Upper bound 
		& 
		$O\!\left(\min\left\{
		\begin{array}{@{}l@{}}
			n+d^{2/3}n^{2/3},\\
			d\sqrt n
		\end{array}
		\right\}\right)$
		& \Cref{thm:ndn_upper_bound}
		& \dndup{}
		& \Cref{thm:dnd_upper_bound} \\
		Unconditional LB 
		& \ndnlowuncond{}
		& \Cref{thm:ndn_unconditional_lowerbound}
		& \dndlowuncond{}
		& \Cref{thm:dnd_unconditional_lowerbound} \\
		Conditional LB
		& \ndnlowcond{}
		& \Cref{thm:ndn_conditional_lowerbound}
		& \dndlowcond{}
		& \Cref{thm:dnd_conditional_lowerbound} \\
		\bottomrule
	\end{tabular}%
		}
	\vspace{6mm}
\end{table*}

\begin{table*}[p]
	\centering
	\caption{Rectangular matrix multiplication over fields in the low-bandwidth model.}
	\label{tab:main-results-fields}
		\resizebox{\textwidth}{!}{%
	\begin{tabular}{lllll}
		\toprule
		& \multicolumn{2}{c}{\ndn{}} & \multicolumn{2}{c}{\dnd{}} \\
		\cmidrule(lr){2-3} \cmidrule(lr){4-5}
		Upper bound 
		& 
		$O\!\left(\min\left\{
		\begin{array}{@{}l@{}}
			n^{2-2/\omega},\\
			n+d^{2/3}n^{2/3},\\
			d\sqrt n
		\end{array}
		\right\}\right)$
		& \Cref{thm:ndn_upper_bound_fields}
		& \dndupfields
		& \Cref{thm:dnd_upper_bound_fields} \\
		Unconditional LB 
		& \ndnlowuncond
		& \Cref{thm:ndn_unconditional_lowerbound}
		& \dndlowuncond
		& \Cref{thm:dnd_unconditional_lowerbound} \\
		Conditional LB
		& $\Omega(\frac{d^{2-4/{3\omega}}}{n^{2/3\omega}}) $
		& \Cref{thm:ndn_conditional_lowerbound_fields}
		& $\Omega(\frac{d^2}{n^{2/\omega}})$
		& \Cref{thm:dnd_conditional_lowerbound_fields}\\
		\bottomrule
	\end{tabular}%
		}
	\vspace{6mm}
\end{table*}

\begin{figure*}[p]
	\centering
	
	\begin{subfigure}{0.49\textwidth}
		\centering
		
		\begin{tikzpicture}[
			scale=4.8,
			axis/.style={thick,->},
			tick/.style={thin},
			]
			\fill[myorange!50]
			(0,0) -- (0,1/2) -- (1/2,1) -- (1,1) -- (1,0) --  cycle;
			\draw[myorange,thick]  (1/2,1) -- (1,1);
			\filldraw[myblue!50, draw= myblue]
			(0,1/2) -- (0,3/2) -- (1,3/2) -- (1,4/3) -- (1/2, 1)-- cycle;
			\draw[dashed] (1/2,1) -- (1,4/3);
			\draw[dashed] (0,1/2) -- (1,3/2);
			\fill[mygrey!40]
			(11/14,1) -- (1,4/3) -- (1,1) -- cycle;
			\draw[dashed] (1/7,0) -- (1,4/3);
			\draw[axis] (0,0) -- (1.1,0) node[above] {$d=n^x$};
			\draw[axis] (0,0) -- (0,1.6) node[above] {$T(n)= n^y$};
			
			\foreach \x/\lab in {0/0, {1/7}/{\frac{1}{7}}, 0.5/{\frac{1}{2}}, 1/1}
			{
				\draw[tick] (\x,0) -- (\x,-0.02) node[below=2pt] {$\lab$};
			}
			
			\draw[tick] (11/14,0) -- (11/14,-0.02)
			node[below=2pt, xshift=-2pt] {$\frac{11}{14}$};
			
			\draw[tick] (5/6,0) -- (5/6,-0.02)
			node[below=2pt, xshift=2pt] {$\frac{5}{6}$};

			\foreach \y/\lab in {0.5/{\frac{1}{2}}, 1/1, 1.333/{\frac{4}{3}}, 1.5/{\frac{3}{2}} }
			{
				\draw[tick] (0,\y) -- (-0.02,\y) node[left=2pt] {$\lab$};
			}
			\draw[dashed] (11/14,0) -- (11/14,1.55);
			\draw[dashed] (5/6,0) -- (5/6,1.55);
			\draw[dashed] (1/2,0) -- (1/2,1.15);
			\draw[dashed] (0,4/3) -- (1.05,4/3);
			\draw[dashed] (0,1) -- (1.05,1);
		\end{tikzpicture}

		\caption{\ndn{} bounds}
		\Description{This is a combination of lower bounds and upper bounds for \ndn{} multiplication}
		\label{fig:ndn}
	\end{subfigure}
	\hfill
	\begin{subfigure}{0.49\textwidth}
		\centering
		
			\begin{tikzpicture}[
				scale=4.8,
			axis/.style={thick,->},
			tick/.style={thin},
			]
			
			\fill[mygrey!40]
			(1/3,0) -- (1,4/3) -- (1,1) -- (1/2,0) -- cycle;
			\draw[dashed] (1/3,0) -- (1,4/3);
			\fill[myorange!50]
			(1/2,0) -- (1,1) -- (1,0) -- cycle;
			\draw[dashed] (1/2,0) -- (1,1);
			
			\fill[myblue!50]
			(0,0) -- (0,3/2) -- (1,3/2) -- (1,4/3) -- cycle;
			
			\draw[axis] (0,0) -- (1.1,0) node[above] {$d=n^x$};
			\draw[axis] (0,0) -- (0,1.6) node[above] {$T(n)= n^y$};
			\foreach \x/\lab in {0/0, 0.333/{\frac{1}{3}}, 0.5/{\frac{1}{2}}, 1/1}
			{
				\draw[tick] (\x,0) -- (\x,-0.02) node[below=2pt] {$\lab$};
			}
			\foreach \y/\lab in {0.5/{\frac{1}{2}}, 1/1, 1.333/{\frac{4}{3}}, 1.5/{\frac{3}{2}} }
			{
				\draw[tick] (0,\y) -- (-0.02,\y) node[left=2pt] {$\lab$};
			}
			\draw[dashed] (0,0) -- (1,4/3);
			\draw[dashed] (0,4/3) -- (1.05,4/3);
			\draw[dashed] (0,1) -- (1.05,1);
			
		\end{tikzpicture}
		
		\caption{\dnd{} bounds}
		\Description{This is a combination of lower bounds and upper bounds for \dnd{} multiplication}
		\label{fig:dnd}
	\end{subfigure}
	
	\caption{This figure summarizes the results from \Cref{tab:main-results}. Here, blue indicates ranges where we have an upper bound, whereas in orange regions we have an unconditional lower bound, and in gray regions we have a conditional lower bound.}
	\Description{This figure summarizes the results from \Cref{tab:main-results}. Here, blue indicates ranges where we have an upper bound, whereas in orange regions we have an unconditional lower bound, and in gray regions we have a conditional lower bound.}
	\label{fig:bounds}
\end{figure*}
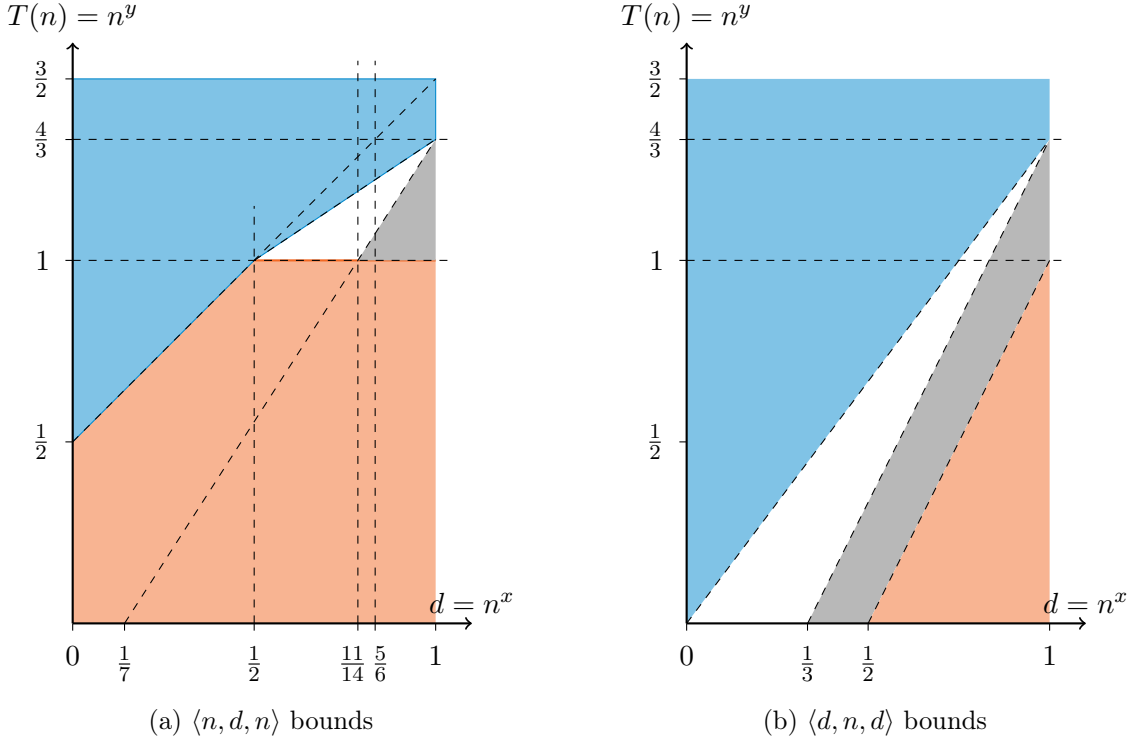

\subsection{Techniques}

\subparagraph{Upper bounds.}
Our upper bounds are based on the familiar idea of subdividing the input and output matrices appropriately into submatrices. However, as we will see, we need complementary approaches for \ndn{} and \dnd{}:
\begin{itemize}
	\item For \ndn{}, we essentially associate one block of the output matrix to one computer. Then it suffices to route the right parts of input to the right computer. Here the key primitive that we need is \emph{broadcasting}, as the same part of the input is needed by multiple computers.
	\item For \dnd{}, we split the input so that we have one instance of \anymatmul{d}{d}{d} per $d$ computers. The final result is the sum of the results of these subproblems, and hence a key primitive that we need in the end is \emph{aggregation}.
\end{itemize}
Here it is notable that such a simple partitioning strategy for \ndn{} can yield an asymptotically optimal algorithm for all $d \le \sqrt{n}$. This is one of our take-home messages: in this region, simple, practical, and easy-to-implement algorithms are provably optimal.

\subparagraph{Lower bounds.}
The technically more interesting part of this work is the lower bounds.

In the \emph{conditional} lower bounds we present reductions that show that fast rectangular matrix multiplication algorithms can be used to design faster algorithms for square matrices. A key technical challenge in the reductions is that, in our setting, matrix dimensions are coupled with the number of computers. For example, in the proof of \cref{thm:ndn_conditional_lowerbound} we therefore split the reduction in two steps:
\begin{enumerate}
	\item We show that if we have a fast algorithm for solving \ndn{} with $n$ computers, we can also solve \kkk{} fast (for some appropriately-chosen $k$) with the \emph{same} number of computers $n$. Note that here the matrix dimensions are different from the number of computers.
	\item Then we show that if we can solve \kkk{} with $n$ computers, we can also solve \anymatmul{N}{N}{N} with $N$ computers (for some appropriately-chosen $N$). Here we are finally in the familiar setting where the matrix dimensions agree with the number of computers.
\end{enumerate}

In the \emph{unconditional} lower bounds, we make use of information-theoretic arguments, where the spirit is this: a fast matrix multiplication algorithm could be used to transmit too many bits between Alice and Bob, in comparison with the total available bandwidth between them. However, here the main technical challenge is that we wanted to show that our lower bounds are not merely artifacts of some specific assumptions related to the input/output distribution.

To this end, we let the algorithm designer choose any fixed division of input and output, and show that the lower bound still holds, as long as the distribution is balanced (of course without the balance requirement the task is trivial: if all the input is in computer 1, and all the output needs to be held by computer 1, there is no need for any communication).

This means that our lower bound has to be adaptive to whatever choice of the input and output distribution the algorithm designer makes. For example, in the proof of \cref{thm:ndn_unconditional_lowerbound} we connect the existence of suitable adversarial strategies with a certain graph coloring problem, and show that the graph coloring problem always admits a solution, and hence we also always have an adversarial strategy, no matter what choices the algorithm designer makes on the distribution of inputs and outputs.

\subsection{Discussion}

\paragraph{Comparison with the MPC setting.}

Rectangular matrix multiplication has been considered in the literature in many different settings. In addition to the prior work in the congested clique model, there is also recent work that has studied matrix multiplication in the MPC (massively parallel computation) model \cite{joshi-deshmukh-etal-2026-matrix-multiplication-in-the}. The MPC setting is rather similar in spirit to the low-bandwidth model; however, the phase transition around $d \approx \sqrt{n}$ in the \ndn{} case seems to be a feature unique to the low-bandwidth model---there are no signs of a similar feature in the MPC results.

\paragraph{Comparison with BSP and distributed-memory matrix multiplication.}

There is a large body of work in the high-performance computing community on the communication complexity of parallel matrix multiplication, in models such as distributed-memory machines, BSP, and BSPRAM~\cite{demmel-eliahu-etal-2013-communication-optimal-parallel,daas-ballard-etal-2022-brief-announcement-tight-memory, irony-toledo-tiskin-2004-communication-lower-bounds-for, ballard-demmel-etal-2012-brief-announcement-strong}. These works are closely related to ours, but they answer a different kind of question. Typically, one fixes the arithmetic circuit of a given algorithm to be evaluated---for example, the conventional cubic matrix-multiplication DAG, or a Strassen-like DAG---and then asks how much data movement is necessary for any parallel implementation that evaluates that circuit.

For conventional rectangular matrix multiplication, Demmel et al.~\cite{demmel-eliahu-etal-2013-communication-optimal-parallel} and Al Daas et al.~\cite{daas-ballard-etal-2022-brief-announcement-tight-memory} give communication-optimal algorithms and matching lower bounds in a distributed-memory model.
Their upper bound results are informative for our setting. If one plugs in $P=n$ processors and dimensions \ndn{}, the resulting bandwidth expressions recover the same two regimes that appear in our upper bounds: $d\sqrt n$ and $d^{2/3}n^{2/3}$. Similarly, for dimensions \dnd{}, the corresponding expression gives $d^{4/3}$. Thus, the upper bounds in communication-avoiding literature can often be easily translated to upper bounds in the low-bandwidth model.

The key difference is that these lower bounds are not lower bounds for the matrix multiplication problem as an arbitrary distributed task. They are lower bounds for evaluating a particular and fixed arithmetic circuit. For example, the lower bounds of Irony et al.~\cite{irony-toledo-tiskin-2004-communication-lower-bounds-for} explicitly apply to \emph{conventional} matrix multiplication: the computation consists of the elementary products $a_{ij}b_{jk}$ and their sums, and the analysis does not apply to Strassen's algorithm or other non-conventional algorithms. Similarly, the lower bounds for Strassen-based algorithms~\cite{ballard-demmel-etal-2012-brief-announcement-strong} fix the Strassen computation DAG, and then prove that any parallel implementation of that DAG must communicate many words. These results still leave open whether a different distributed algorithm could solve the same problem faster.

Our lower bounds do not assume that the algorithm evaluates a prescribed arithmetic circuit. Local computation is unrestricted, and the algorithm may use any strategy as long as, after the communication rounds, machines collectively output the product matrix. In particular, our lower bounds also account for techniques such as \emph{communication by silence}, where information can be encoded or derived not only in the messages that are sent, but also in the absence of messages~\cite{dietzfelbinger-kutylowski-reischuk-1994-exact-lower,roughgarden-vassilvitskii-wang-2018-shuffles-and}. Moreover, our bounds already hold for the seemingly easier special case of matrix multiplication over the Boolean semiring $(\{0,1\},\lor,\land)$.
In this sense, our lower bound for \ndn{} with $d\le \sqrt n$ shows that the $d\sqrt n$ barrier suggested by communication-avoiding conventional matrix multiplication is not merely a limitation of such algorithms: in the low-bandwidth model it is an unconditional barrier for all algorithms.

Upper bounds in BSP and BSPRAM also fit naturally into this picture. McColl and Tiskin~\cite{mccoll-tiskin-1999-memory-efficient-matrix} give BSPRAM algorithms for matrix multiplication. If one focuses on the bandwidth term and sets the number of processors to $p=n$, these results recover the usual baseline bounds $O(n^{4/3})$ for standard semiring multiplication and $O(n^{2-2/\omega})$ for fast matrix multiplication over fields; these are consistent with the square-matrix primitives that we use as black boxes.

\paragraph{Comparison with the supported setting.}

We point out that prior work on sparse matrix multiplication in the low-bandwidth setting
\cite{gupta-hirvonen-etal-2022-sparse-matrix-multiplication,gupta-korhonen-etal-2025-low-bandwidth-matrix} uses the \emph{supported} version of the model: the sparsity structure of the matrices is known in advance, and only the exact values of those elements are known at runtime. We emphasize that in the present work, we do not need to make any such structural assumptions, as our matrices are dense.

\subsection{Open Problems for Future Work}

The landscape of \ndn{} is now fully understood for $d \le \sqrt{n}$. On the other hand, proving a tight unconditional bound for $d \ge \sqrt{n}$ would likely require major breakthroughs in our understanding of matrix multiplication in general. However, we believe that proving a tight \emph{conditional} lower bound for $d \ge \sqrt{n}$ (or at least closing \emph{some} of the gap) is within the reach of current techniques.

The major open question is understanding the landscape of \dnd{}, where we currently have wide gaps for small values of $d$. We are not aware of fundamental barriers there. A particularly pressing question is understanding the case of $d \approx n^{1/3}$ for \dnd{}, which is the point where we do not have any lower bounds, conditional or unconditional.

%% file: preliminaries.tex
\section{Preliminaries}\label{sec:prelim}

This section contains the definition of the model and problem, as well as a selection of helpful results and observations in the low-bandwidth model.

\subsection{Model and Problem Definitions} \label{subsec:definitions}

We work in the \emph{low-bandwidth model}, a distributed system that consists of $n$ machines (or nodes), denoted by
$M = \{1,2,\dots,n\}$.
Computation proceeds in synchronous communication rounds.

\subparagraph{Local computation.}
Each machine has unbounded local computational power and unbounded local memory. Local computation is free; the complexity measure of interest is the number of communication rounds.

\subparagraph{Communication.}
In each round, every machine may
\begin{itemize}
	\item send exactly one message to another machine, and
	\item receive exactly one message from another machine.
\end{itemize}
Each message contains at most $O(\log n)$ bits.
Messages sent in round $r$ are received by the end of the same round.
Equivalently, in each round the communication forms a directed graph of maximum in-degree and out-degree~$1$.

\subparagraph{Problem.} In this work, we study rectangular matrix multiplication. Given matrices $A,B$, the task is to compute the product $X=AB$. We denote by \anymatmul{a}{b}{c} the task of multiplying an $a\times b$ matrix $A$ by a $b\times c$ matrix $B$ in the low-bandwidth model. In this work, we will focus on \ndn{} and \dnd{} multiplications, where $d\in [n]$.

Throughout the paper, we perform matrix multiplication over an arbitrary semiring $(\mathcal{S}, +, \cdot)$. 
We assume that every element of $\mathcal{S}$ admits an encoding of size $O(\log n)$ bits. 
In particular, a single semiring element fits within one communication message in the low-bandwidth model.

\subparagraph{Input and output distribution.}
The input is evenly distributed among the $n$ machines, and each machine is responsible for producing an evenly sized portion of the output.

More concretely, for the instances considered in this work:
\begin{itemize}
	\item each machine initially holds at most $O(d)$ input entries, 
	\item each machine is required to output at most $O(d)$ entries of \dnd{} instances, or $O(n)$ entries of \ndn{} instances.
\end{itemize}

The assignment of input and output indices to machines is fixed in advance and known to all machines (however, the values stored in the entries are not). In particular, machines know the value of $d$ and $n$, and which multiplication instance is given.

\subparagraph{Complexity.}
An algorithm runs in $T(n)$ rounds if, after $T(n)$ communication
rounds, all machines have produced their assigned output.

\subsection{Observations and Toolbox}\label{subsec:toolbox}
 	
 For any integer $n>0$, we denote by $[n]$ the set $\{1,\ldots, n\}$.

\begin{theorem}[Square matrix multiplication~{\cite{censor-hillel-kaski-etal-2019-algebraic-methods-in-the}, \cite{gupta-hirvonen-etal-2022-sparse-matrix-multiplication}}]
	\label{thm:dense-mm-known}
	Given $n \times n$ matrices $A,B$ in the low-bandwidth model, $X=AB$ can be computed for semirings in $O(n^{4/3})$ rounds, or in $O(n^{2-2/\omega})$ rounds for fields.
\end{theorem}
\begin{proof}
	Follows from simulating each congested clique round of the $O(n^{1/3})$-round semiring algorithm (currently best known) from \cite{censor-hillel-kaski-etal-2019-algebraic-methods-in-the} and their $O(n^{1-2/\omega})$-round algorithm over fields in $O(n)$ rounds of the low-bandwidth model. 
\end{proof}

Similarly, by simulating the $O(d^{2/3}/n^{1/3}+1)$-round result from \cite{censor-hillel-leitersdorf-turner-2018-sparse-matrix} we have the following.
\begin{theorem}[Simulation of~{\cite{censor-hillel-leitersdorf-turner-2018-sparse-matrix}}]
	\label{thm:ndn_pastwork_upperbound}
	In the low-bandwidth model, there is an $O(d^{2/3}n^{2/3}+n)$-round algorithm that solves \ndn{} over semirings.
\end{theorem}

Next, we briefly mention some useful tools. An important note here is that machines do not coordinate during the algorithm to know to which machine they need to send a message (or from which machine they need to receive one); this scheduling is already hard-coded into the algorithm and execution is merely the step in which the values are revealed and transferred. This is done to avoid the costly runtime coordination overhead. We will see this pattern in our algorithms as well.

\begin{lemma}[Broadcast~{\cite[Lemma 6.9]{gupta-korhonen-etal-2025-low-bandwidth-matrix}}]
	 \label{lem:broadcast}
	In the low-bandwidth model, broadcasting a single bit $v \in \{0,1\}$ to all machines takes $\Theta(\log n)$ rounds. Moreover, if $k>1$ messages are given, broadcasting can be done in $O(k+ \log n)$ rounds. 
\end{lemma}
\begin{proof}
Broadcasting can be done in a power-jumping manner. Suppose machine~$1$ initially holds the value $v$. In round $i$, there are $2^{i-1}$ machines that know $v$. Each of them sends $v$ to a distinct machine that does not yet know it,
for example to machines $\{2^{i-1}+1, \ldots, 2^{i}\}$. Thus after round $i$, $2^i$ machines know $v$. After $O(\log n)$ rounds, all machines know $v$. The argument works for any value $v$ that fits in a single message. By carefully interleaving broadcast instances of the above form~\cite{kannan-ballard-park-2016-a-high-performance-parallel}, one can schedule all the messages in a pipeline manner to achieve $O(k+ \log n)$ rounds.

For details of the lower bound, we refer the reader to \cite[Lemma 6.9]{gupta-korhonen-etal-2025-low-bandwidth-matrix}.
\end{proof}

\begin{lemma}[Aggregation~{\cite[Lemma 6.4]{gupta-korhonen-etal-2025-low-bandwidth-matrix}}] \label{lem:aggregation} 
	In the low-bandwidth model where all $n$ machines are given a single bit in $\{0,1\}$, computing logical AND, OR, and sum of the values takes $\Theta(\log n)$ rounds. Moreover, if machines are instead given $k$ values that fit in a message, aggregation can be done in $O(k + \log n)$ rounds.
\end{lemma}
\begin{proof}
	The idea is similar to \Cref{lem:broadcast}, but done in the reverse order. Similarly, this works for any $v$ that fits into a message, and can be generalized to $k$ values by the same pipelining idea of \cite{kannan-ballard-park-2016-a-high-performance-parallel}.
	
	For details of the lower bound, we refer the reader to \cite[Lemma 6.4]{gupta-korhonen-etal-2025-low-bandwidth-matrix}.
\end{proof}

  We occasionally refer to a \emph{distribution} $D$ of data to emphasize which machine $m \in M$ holds which value $v$ (e.g. an element of a matrix); here $v$ is typically small enough to fit in a single message, i.e. $O(\log n)$ bits.
 
 \begin{observation}[Redistribution] \label{obs:redistribution}
 	Given a distribution $D$ of $d$ values per machine, we can rearrange the values to any distribution $D'$ of $d$ values per machine. This is done thanks to the fact that the bipartite graph of senders and receivers where a value's home and destination machine form an edge is $d$-regular, and therefore a $d$-coloring of the edges exists. One simply iterates over color classes and sends all messages of a color class in one round. Computing the coloring requires no communication and can be done locally by all machines as long as the distributions $D$ and $D'$ are known globally. 
 \end{observation}

%% file: ndn.tex
\section{\texorpdfstring{\boldmath\ndn{}}{⟨n,d,n⟩} Multiplication}\label{sec:n_d_n}

\subsection{Upper Bound for \texorpdfstring{\boldmath\ndn{}}{⟨n,d,n⟩}}\label{subsec:ndn_upper_bound}

\begin{theorem}\label{thm:ndn_upper_bound}
	In the low-bandwidth model with $n$ machines, 
	$\ndn{}$ multiplication over semirings can be computed in $\ndnup{}$ communication rounds.
\end{theorem}
\begin{proof}
	Let $A,B$ be given matrices, and we are computing $X=AB$. Let $A$ be an $n \times d$ matrix and $B$ a $d \times n$ matrix.
Assume that $n$ is a perfect square; otherwise, we safely pad the matrices with zeros.
We partition $X$ into $\sqrt{n} \times \sqrt{n}$ blocks and write
\[
X =
\begin{bmatrix}
	\chi_{1,1} & \chi_{1,2} & \cdots & \chi_{1,\sqrt{n}} \\
	\chi_{2,1} & \chi_{2,2} & \cdots & \chi_{2,\sqrt{n}} \\
	\vdots & \vdots & \ddots & \vdots \\
	\chi_{\sqrt{n},1} & \chi_{\sqrt{n},2} & \cdots & \chi_{\sqrt{n},\sqrt{n}}
\end{bmatrix}.
\]
Formally, for every \( i,j \in [\sqrt{n}] \),
\[
(\chi_{i,j})_{p,q}
=
X_{(i-1)\sqrt{n} + p,\; (j-1)\sqrt{n} + q}
\quad
\text{for all } p,q \in [\sqrt{n}].
\]
We also partition $A$ into $\sqrt{n}$ row blocks:
\[
A =
\begin{bmatrix}
	A_{1} \\
	A_{2} \\
	\vdots \\
	A_{\sqrt{n}}
\end{bmatrix}.
\]

Similarly, we partition $B$ into $\sqrt{n}$ column blocks:
\[
B =
\begin{bmatrix}
	B_{1} & B_{2} & \cdots & B_{\sqrt{n}}
\end{bmatrix}.
\]
Rewriting the multiplication, we have
\begin{align*}
	AB &=
	\begin{bmatrix}
		A_{1} \\
		A_{2} \\
		\vdots \\
		A_{\sqrt{n}}
	\end{bmatrix}
	\begin{bmatrix}
		B_{1} & B_{2} & \cdots & B_{\sqrt{n}}
	\end{bmatrix}
	\\ &=
	\begin{bmatrix}
		A_{1} B_{1} & A_{1} B_{2} & \cdots & A_{1} B_{\sqrt{n}} \\
		A_{2} B_{1} & A_{2} B_{2} & \cdots & A_{2} B_{\sqrt{n}} \\
		\vdots & \vdots & \ddots & \vdots \\
		A_{\sqrt{n}} B_{1} & A_{\sqrt{n}} B_{2} & \cdots & A_{\sqrt{n}} B_{\sqrt{n}}
	\end{bmatrix}
	=
	X,
\end{align*}
where for every \( i,j \in [\sqrt{n}] \),
$\chi_{i,j} = A_i B_j$.

For simplicity, we assume that each machine initially holds a row of $A$ and a column of $B$; otherwise, we can redistribute in $O(d)$ rounds by \Cref{obs:redistribution}. Moreover, for every pair $i,j \in [\sqrt{n}]$ we assign a distinct machine, denoted by $M_{i,j}$, responsible for outputting $\chi_{i,j}$.
Since $\chi_{i,j}= A_i B_j$, it is enough to route all elements of $A_i$ and $B_j$ to $M_{i,j}$, then compute $A_i B_j$ locally.
\subparagraph{Routing.} We first route all elements of $A$ to the right machines. For any $i \in [\sqrt{n}]$, $A_i$ needs to be sent to $\{M_{i,1}, M_{i,2}, \ldots, M_{i,\sqrt{n}} \}$. For all $A_i$'s, in $O(d\sqrt{n})$ rounds we can send $A_i$ from host machines to $M_{i,1}$. Then, we can use the power jump technique from \Cref{lem:broadcast} to broadcast $A_i$ to all the above machines in $O(d\sqrt{n}+\log n)=O(d\sqrt{n})$ rounds. Note that since no machine belongs to more than one broadcast instance, all instances can run in parallel. Similarly, we route and broadcast $B_j$ to all machines $\{M_{1,j}, M_{2,j}, \ldots, M_{\sqrt{n},j} \}$ in $O(d\sqrt{n})$ rounds.

Furthermore, \ndn{} can be solved by \Cref{thm:ndn_pastwork_upperbound} in
$
O(d^{2/3}n^{2/3}+n)
$
rounds. Combining the two algorithms, we conclude that \ndn{} can be solved in
$
\ndnup
$
rounds.
\end{proof}

Note that our $O(d\sqrt{n})$-round algorithm outperforms known algorithms up to $d=O(\sqn)$. 
We can similarly prove the following for multiplication over fields.

\begin{theorem}\label{thm:ndn_upper_bound_fields}
	In the low-bandwidth model with $n$ machines, 
	$\ndn{}$ multiplication over fields can be computed in $\ndnupfields{}$ communication rounds.
\end{theorem}

\subsection{Unconditional Lower Bound for \texorpdfstring{\boldmath\ndn{}}{⟨n,d,n⟩}} \label{subsec:n_d_n_lower_bound}

Suppose $AB = X$ is an \ndn{} instance. Let $Y$ be an $n \times n$ matrix which is the same as $X$, except it only contains those elements that machine $1$ will output, in the same position as they appear in $X$, and the other elements are $0$. Since each machine is responsible for outputting an even share of the product matrix, we know that $Y$ contains $n$ nonzero elements. Our idea to prove the lower bound works as follows. Suppose matrix $A$ contains arbitrary elements but we can set the elements of $B$ as $0$ and $1$ in such a way that after multiplying $A$ with $B$, $Y$ contains $k$ elements of $A$, in other words $y_{ij} = a_{ij}$ for some $k$ indices (where $y_{ij}$ and $a_{ij}$ are the elements of $Y$ and $A$). Since machine $1$ initially contains only $d$ elements of $A$, we will prove that any algorithm will take $\Omega(k-d)$ rounds to communicate the remaining $ \geq (k-d)$ elements of $A$ from other machines to machine $1$ (or we might use the opposite setting where elements of $A$ are set to $0$ and $1$ and matrix $B$ is any arbitrary matrix). We first prove the following lemma, which we will use to prove the desired lower bound.

Suppose Alice and Bob are two players such that Alice has a bit vector of length $b \log n$ that Bob wants to learn. Let $\mathcal{A}$ be an algorithm for this task in the low-bandwidth model.
\begin{lemma}\label{lemma:vec_comm}
Algorithm $\mathcal{A}$ requires at least $b$ rounds.
\end{lemma}

\begin{proof}
There are $2^{b\log n}$ bit vectors of length $b\log n$; enumerate them as $v_1,\dots,v_{2^{b\log n}}$. Assume that $\mathcal{A}$ completes the communication in $r<b$ rounds. Since in each round Bob can receive at most $\log n$ bits, he receives at most $r \log n$ bits in total. Because $r \log n < b \log n$, there must exist two distinct vectors $v_i$ and $v_j$ such that Bob receives the same sequence of bits when 
Alice holds either $v_i$ or $v_j$. Consequently, Bob cannot distinguish between these two cases and therefore outputs an incorrect vector in at least one case. Hence, $\mathcal{A}$ must take at least $b$ rounds.
\end{proof}

Returning to our lower bound proof, we divide the rows of $Y$ into two parts, heavy and light: the heavy part contains those rows that have at least $\sqrt{n}$ elements, and the light part contains those rows that have fewer than $\sqrt{n}$ elements. Without loss of generality, the rows in the heavy part have higher indices (top rows) than the rows in the light part (since we can swap rows of $A$ in a way such that the condition holds. In fact, we can assume that the rows are sorted, from top to bottom, according to the number of elements in them). 
Now divide the remaining analysis into two cases.

\subparagraph{Case 1: Heavy rows contain at least $\frac{n}{2}$ elements.} If the number of heavy rows is $\geq d$, then the first $d$ top rows (heavy rows with high indices) collectively contain at least $d\sqrt {n}$ elements. Also, if the number of heavy rows is less than $d$, then we know that the first $d$ top rows contain $\Omega(n)$ elements. Now, suppose that we construct a matrix $B$ in such a way that it contains some arbitrary elements in the same positions as in the first $d$ top  rows of $Y$. Matrix $A$ contains the identity matrix in its first $d$ top rows and $d$ columns, and the remaining elements are $0$. Now, in order to do $A \times B$, these $\min \{d\sqrt {n},n\}$ elements of $B$ need to be communicated to machine $1$. We treat machine $1$ as Bob and all remaining machines together as Alice. Alice initially holds $\min \{d\sqrt {n},n\}$ elements of $B$, corresponding to a bit vector of length 
$(\min \{d\sqrt {n},n\}) \log n$ that Bob must learn. By \cref{lemma:vec_comm}, any algorithm requires \ndnlowuncond{} rounds for the multiplication of $A$ and $B$ in this case.

	\subparagraph{Case 2: Light rows contain more than $\frac{n}{2}$ elements.} Since $d \leq n$, we can say that there are at least $\Omega(\sqn)$ light rows. Mark the first $d$ elements from each light row of $Y$ (if the number of elements in a row is less than $d$, then select all of them). We can prove that the number of marked elements is $\Omega(d\sqn)$. The proof goes as follows. First, let us sort the rows according to the number of elements present in them. Look at the top $\sqn /2$ light rows. (i) If the last row among these $\sqn/2$ rows contains at least $d$ elements, then the marked elements in the top $\sqn/2$ rows alone give us $\Omega(d \sqn)$ elements. (ii) If the last row among these $\sqn/2$ rows does not contain $d$ elements, that means all the remaining light rows also contain fewer than $d$ elements. Now, let us count the maximum number of unmarked elements in this case. All the unmarked elements can only be present in the top $\sqn/2$ rows, therefore they are at most $\frac{\sqn}{2} (\sqn -d)$. By subtracting that from the total elements of the light rows, we get that the total number of marked elements is $\Omega(d\sqn)$. 

Now that we have marked $\Omega(d \sqn)$ elements in matrix $Y$, we will show that for any arbitrary matrix $A$, one can carefully set the elements of $B$ as $0$ and $1$ such that after multiplication, $\Omega(d \sqn)$ elements of $A$ appear at these marked positions of $Y$. Note that if each column of $B$ contains exactly one $1$ and all the other elements are $0$, then each column of $X$ will be nothing but a column of $A$. Now the idea is to set elements of $B$ in such a way that each column contains a single $1$ and after multiplication $\Omega(d\sqn)$ distinct elements of $A$ land at the marked positions of $Y$. Let us suppose that the $i$-th row of $Y$ contains $t$ marked elements, namely $y_{i,j_1}, y_{i,j_2}, \ldots, y_{i,j_t}$. Also, suppose in matrix $B$ columns $j_1,j_2, \ldots, j_t$ contain $1$ in rows $k_1,k_2, \ldots, k_t$, respectively (and of course all the other elements in these columns are $0$). In such a setting, after multiplication we will have that $y_{i,j_l} = a_{i,k_l}$ for all $l \in [t]$. Notice that if all $k_l$ are distinct, then $a_{i,k_l}$ is also distinct. In other words, we can say that if there are $p$ distinct $k_l$ among $k_1,k_2, \ldots, k_t$, then $p$ distinct elements from row $i$ of $A$ will land on the marked elements of matrix $Y$. For each row $i$ of matrix $Y$, we define a number $f_i$ as follows. If $y_{i,j_1}, y_{i,j_2}, \ldots, y_{i,j_t}$ are the marked elements of row $i$ of $X$, then $f_i$ is the maximum number of columns among $j_1,j_2, \ldots, j_t$ in matrix $B$ such that no two columns among these $f_i$ columns contain $1$ in the same row. Now if we can prove that we can set the matrix $B$ in such a way that $\Sigma f_i = \Omega(d \sqn)$ then we are done.

We can convert this problem into a graph coloring problem as follows. Let us look at the matrix $Y$ (only marked elements) as a bipartite graph $G=(L,R)$ with $n$ nodes on each set $L$ and $R$. Let us index the nodes on each side as $1,2,\ldots,n$. Put an edge $(i,j)$ if $y_{i,j}$ is marked. In other words, edges incident on node $i \in L$ represent the marked elements of row $i$ in $Y$ and edges incident on node $j \in R$ represent the marked elements in column $j$ of $Y$.

\begin{lemma}[Distinct-Neighbor Coloring] \label{distinct-coloring}
	Let $G=(L,R,E)$ be a bipartite graph with $|L|=|R|=n$.
	Assume that every vertex $v \in L$ has degree at most $d\leq n$.
	Then there exists a coloring $c : R \to [d]$ such that
	\[
	\sum_{v \in L} |c(N(v))| \ge \frac{|E|}{2}.
	\]
\end{lemma}

\begin{proof}
	Let $N(v)$ be the neighbors of $v$. Assign to each $j \in R$ a color $c(j) \in [d]$
	independently and uniformly at random.
	
	Fix a vertex $i \in L$ with degree $t \le d$.
	For a fixed color $z \in [d]$, the probability that
	none of the $t$ neighbors of $i$ receives color $z$ is
	$(1 - 1/d)^t$.
	Hence the probability that color $z$ appears among the
	neighbors of $i$ is
	$1 - (1 - 1/d)^t$.
	By linearity of expectation,
	\[ \mathbb{E}[|c(N(i))|] = d(1 - (1 - 1/d)^t).\]
	Using the inequality
	\[ d (1 - (1 - 1/d)^t) \ge \frac{t}{2} \quad \text{for } t \le d,\]
	we obtain
	\[\mathbb{E}[|c(N(i))|] \ge \frac{t}{2}.\]
	Summing over all $i \in L$ and using linearity of expectation,
	\[
	\mathbb{E}\left[\sum_{i \in L} |c(N(i))|\right] \ge \sum_{i \in L} \frac{\deg(i)}{2}=\frac{|E|}{2}.
	\]
	Therefore there exists a coloring $c$
	such that
	\[
	\sum_{i \in L} |c(N(i))| \ge \frac{|E|}{2}. \qedhere
	\]
\end{proof}

Now this lemma gives us the combinatorial gadget that enables us to force a lot of distinct entries of $A$ to appear in the positions that machine $1$ is responsible for outputting. We now construct the matrix $B$ as follows: each column $j$ contains a single $1$ in row $c(j)$, and all other entries are $0$.
Then for all $(i,j)\in E$,
\[ x_{i,j} = \sum_{t=1}^d a_{i,t} b_{t,j} = a_{i,c(j)}.
\]

Therefore, for each row $i$ the number of distinct entries of $A$ that appear in the marked positions of row $i$ is exactly $|c(N(i))|$.
Summing over all rows, the total number of distinct entries of $A$ that appear in the marked positions is at least
\[ \sum_{i\in L} |c(N(i))| = \Omega(d\sqrt{n}).
\]

Machine $1$ must output all $\Omega(d \sqn)$ distinct entries corresponding to $Y$; by \Cref{lemma:vec_comm}, this takes $\Omega(d \sqn)$ rounds. On the other hand, by definition, machine $1$ does not output more than $n$ elements. Thus, the \ndnlowuncond{} bound follows.

This proves the following theorem.

\begin{theorem} \label{thm:ndn_unconditional_lowerbound}
	In the low-bandwidth model with $n$ machines, any algorithm that computes the product \ndn{} over fields or semirings requires $\ndnlowuncond$ communication rounds.
\end{theorem}

 In particular, we get the following corollary.
 \begin{corollary}\label{cor:nnn_unconditional_lowerbound}
 	In the low-bandwidth model with $n$ machines, any algorithm that computes the product \nnn over fields or semirings requires $\Omega(n)$ communication rounds.
 \end{corollary}

\subsection{Conditional Lower Bound for \texorpdfstring{\boldmath\ndn{}}{⟨n,d,n⟩}}
In this section, we will prove $\Omega\!\left(\frac{d^{14/9}}{n^{2/9}}\right)$ for \ndn{} assuming that $\Omega(n^{\frac{4}{3}})$ is the lower bound for \nnn{}. We begin by proving that if we have an algorithm to solve \ndn{} instances, we can use it to solve \kkk{} instances, with no asymptotic overhead, where $k = (nd^2)^{\frac{1}{3}}$.

\begin{lemma}
	\label{lemma:rect-to-square}
	If \ndn{} can be solved in $O(n^\beta)$ rounds, then we can solve \kkk{} in $O(n^\beta)$ rounds as well.
\end{lemma}

\begin{proof}
	Let $A\times B = C$ be an instance where $A, B,$ and $C$ are of size $k\times k$. We can divide the task of multiplying $A$ and $B$ into multiplying $(\frac{n}{d})$ instances of \ddd{}. Let us first divide the matrices $A$ and $B$ into $(\frac{n}{d})^{\frac{2}{3}}$ submatrices of size $d \times d$ as follows. Let $t = (\frac{n}{d})^{\frac{1}{3}}$.
   
    \[
    \begin{bmatrix}
	A_1 & A_2 & \cdots &  A_t\\
    A_{t+1} & A_{t+2} & \cdots & A_{2t}\\
   \vdots & \vdots & \ddots & \vdots \\
    A_{t^2-t+1} & A_{t^2-t+2} & \cdots &A_{t^2}\\
	\end{bmatrix} \times 
    \begin{bmatrix}
	B_1 & B_2 & \cdots &  B_t\\
    B_{t+1} & B_{t+2} & \cdots & B_{2t}\\
   \vdots & \vdots & \ddots & \vdots \\
    B_{t^2-t+1} & B_{t^2-t+2} & \cdots & B_{t^2}\\
	\end{bmatrix}.
    \]
     From this, we conclude that in order to compute $A \times B$, we need to compute $t^3 = \frac{n}{d}$ instances of type $A_i \times B_j$ for some combination of $i,j \in [t^2]$. Let us rename these $t^3$ instances as $A_i' \times B_i'$. Now we prove that if \ndn{} can be solved in $O(n^{\beta})$ rounds, then these $t^3$ instances can also be computed in $O(n^\beta)$ rounds. We create an \ndn{} instance by arranging the matrices
     $A_i'$ vertically and the matrices $B_i'$ horizontally as follows:
     \[
	\begin{bmatrix}
		A_1' \\
		A_2' \\
		\vdots \\
		A_{t^3}'
	\end{bmatrix} \times 
     [\, B_1' \quad B_2' \cdots B_{t^3}' \,] =
     \begin{bmatrix}
	C_1' & \cdots & \cdots &  \cdots \\
     & C_{2}' & \cdots & \cdots \\
   \vdots & \vdots & \ddots & \vdots \\
	    \cdots & \cdots & \cdots &C_{t^3}'\\
		\end{bmatrix}
		\]
	    Let us denote these larger matrices by $A', B',$ and $C'$. Let \[C_1',C_2', \cdots, C_{t^3}'\] be the $d \times d$ matrices that appear on the diagonal of the resultant matrix $C'$. We can see that the matrices $C_i'$ are the desired $t^3$ matrices that we wanted, which, by our assumption, can be computed in $O(n^\beta)$ rounds. This implies that $C$ can also be computed in $O(n^\beta)$ rounds.
\end{proof}

Let us assume that $d = n^\alpha$ for some $0 \leq \alpha \leq 1$. From the above, we conclude that $\langle n^{\frac{1 +2 \alpha}{3}}, n^{\frac{1 +2 \alpha}{3}}, n^{\frac{1 +2 \alpha}{3}}\rangle$ can be solved in $O(n^\beta)$ rounds
	 with $n$ machines. Let us denote $\gamma ={(\frac{1 +2 \alpha}{3})}$ and $N = n^\gamma $. That implies that \ensuremath{\langle N, N, N\rangle} can be solved in $O(N^{\frac{\beta}{\gamma}})$ rounds with $N^{\frac{1}{\gamma}}$ machines. Since $\gamma \leq 1$, we have $N^{\frac{1}{\gamma}} \geq N$. Now suppose that instead of $N^{\frac{1}{\gamma}}$ machines, we use $N$ machines to run the same algorithm. In such a scenario, one machine is responsible for handling the communication of $N^{\frac{1}{\gamma}-1}$ machines. This means that a task that could be done in one round with $N^{\frac{1}{\gamma}}$ machines will now require $N^{\frac{1}{\gamma}-1}$ rounds. Therefore, we can say that \ensuremath{\langle N, N, N\rangle} can be solved in $O(N^{\frac{\beta}{\gamma}+\frac{1}{\gamma}-1})$ rounds with $N$ machines. Now if we assume that \ensuremath{\langle n, n, n\rangle} requires $\Omega(n^{\frac{4}{3}})$ rounds with $n$ machines in the low bandwidth model, from the above we can conclude that $\frac{\beta + 1 - \gamma}{\gamma} \geq \frac{4}{3}$. This implies \[\beta \geq \frac{7 \gamma}{3} - 1 = \frac{7(2\alpha + 1)}{9} - 1 = \frac{14 \alpha - 2}{9}.\] Since $d = n^\alpha$, we can conclude that \ndn{} requires $\Omega\!\left(\frac{d^{14/9}}{n^{2/9}}\right)$ rounds.

\begin{theorem} \label{thm:ndn_conditional_lowerbound}
	Assume that we can solve \ndn{} over semirings in $O(\frac{d^{14/9}}{n^{2/9+\epsilon}})$ rounds in the low-bandwidth model with $n$ machines. Then we can also solve dense matrix multiplication in $O(n^{4/3-\epsilon/2})$ rounds.
\end{theorem}

By using similar ideas with small adjustments we can also prove the following bound for multiplication over fields.

\begin{theorem}\label{thm:ndn_conditional_lowerbound_fields}
	Assume that we can solve \ndn{} over fields in
	$O(\frac{d^{2-4/{3\omega}}}{n^{2/3\omega+\epsilon}}) $ rounds, for some $\epsilon>0$, in the low-bandwidth model with $n$ machines.
	Then we can solve dense $n\times n$ matrix multiplication over fields in $O(n^{2-\frac{2}{\omega}-\epsilon/2})$ rounds.
\end{theorem}

%% file: dnd.tex
\section{\texorpdfstring{\boldmath\dnd{}}{⟨d,n,d⟩} Multiplication}\label{sec:d_n_d}

 \subsection{Upper Bound for \texorpdfstring{\boldmath\dnd{}}{⟨d,n,d⟩}}\label{subsec:dnd_upper_bound}

\begin{theorem}  \label{thm:dnd_upper_bound}
	In the low-bandwidth model with $n$ machines, $\dnd{}$ multiplication over semirings can be computed in \dndup{} communication rounds.
\end{theorem}
\begin{proof}
	Let $AB=X$ be the multiplication we are computing, where $A$ is a $d \times n$ matrix and $B$ is an $n \times d$ matrix.
	
	We split matrix $A$ into $m:= \frac{n}{d}$ square submatrices $A_1, \ldots, A_{m}$ of size $d \times d$ such that \[ A = [\, A_1 \mid A_2 \mid \cdots \mid A_m \,]. \] Similarly, we split $B$, \[
	B=
	\begin{bmatrix}
		B_1 \\
		B_2 \\
		\vdots \\
		B_m
	\end{bmatrix}.
	\] Note that if $n$ is not divisible by $d$, we can safely pad the last block with zeros.
	Now $X= AB$ can be computed as $X= \sum_{i=1}^{m} A_i B_i$, so it is enough to compute the smaller dense but square matrix multiplication instances and add them all together. 
	
	\subparagraph{Step 1: Multiplications.}
	We assume that for each $i \in [m]$ there are $d$ machines such that each machine holds a row of $A_i$ and a column of $B_i$, otherwise a redistribution of the input can easily be done in $O(d)$ rounds by \Cref{obs:redistribution}. This essentially gives us $m$ independent instances of \anymatmul{d}{d}{d}, i.e. $C^{(i)}=A_i B_i$, using $d$ machines. These can be solved in parallel in $O(d^{4/3})$ rounds by \Cref{thm:dense-mm-known}.
	\subparagraph{Step 2: Aggregation.}
	To compute $X$, we just need to aggregate all the values to compute the sum $X= \sum_{i=1}^{m} C^{(i)}$.
		Formally, we may assume that for each result matrix $C^{(i)}$, where $i \in [m]$, each row $C^{(i)}_{j,*}$ is held in a single machine $M^{(i)}_j$. Then using \Cref{lem:aggregation} we can aggregate the sum in $O(d + \log n)$ rounds.
	
	Combining the above two steps, we solve \dnd{} multiplication in \dndup{} communication rounds.
\end{proof}

	Similarly, we can show the following result.
	\begin{theorem}  \label{thm:dnd_upper_bound_fields}
		In the low-bandwidth model with $n$ machines, $\dnd{}$ multiplication over fields can be computed in \dndupfields{} communication rounds.
	\end{theorem}

\subsection{Conditional Lower Bound for \texorpdfstring{\boldmath\dnd{}}{⟨d,n,d⟩}} \label{subsec:dnd-lowerbound-conditional}

\begin{theorem} \label{thm:dnd_conditional_lowerbound}
	Assume that we can solve \dnd{} over semirings in $O(\frac{d^2}{n^{2/3+\epsilon}})$ rounds, $\epsilon>0$, in the low-bandwidth model with $n$ machines. Then we can also solve dense matrix multiplication in $O(n^{4/3-\epsilon/2})$ rounds.
\end{theorem}
\begin{proof}
Let $d = n^\alpha$, where $0 < \alpha \leq 1$. Suppose we have an algorithm that solves $\dnd{} \equiv \anymatmul{n^\alpha}{n}{n^\alpha}$ type of matrix multiplication instance in $O(n^\beta)$ rounds. We can use this algorithm as a black box to solve \anymatmul{n}{n}{n} as follows.  
Given matrices $A,B \in \mathcal{S}^{n \times n}$, to compute $A \times B = C$, we divide the matrices $A,B$ and $C$ into $\na \times \na$ square submatrices $A_{ij}, B_{ij}$, and $C_{ij}$ where $i,j \in [n^{1-\alpha}]$ such that $C_{ij} = \sum_{k=1}^{n^{1-\alpha}} A_{ik} \times B_{kj}$. 
Note that $C_{ij}$ can be thought of as an $\anymatmul{\na}{n}{\na}$ instance (where matrices $A_{ik}$ are arranged horizontally to create an $n^\alpha \times n$ matrix and matrices $B_{kj}$ are arranged vertically to form an $n \times n^\alpha$ matrix). Therefore, we can use the above $n^\beta$ round algorithm to compute this product. Applying the same algorithm to all $n^{2-2\alpha}$ instances of $C_{ij}$, one by one, gives us an $O(n^{2-2\alpha +\beta})$ round algorithm to compute $C$. 

Now, if there are no algorithms that solve square matrix multiplication in $O(n^{4/3-\epsilon'})$ rounds, for some $\epsilon'>0$, then ${2-2\alpha +\beta} \geq 4/3 - \epsilon'$ and hence
$\beta \geq 2 \alpha - \frac{2}{3}-\epsilon'$. Recall that we set $d=\na$. By substitution, it follows that no $O(d^2 n^{-\frac{2}{3} - \epsilon})$-round algorithm for \dnd{} exists, where $\epsilon = 2\epsilon'$.
\end{proof}

By using similar ideas with small adjustments we can also prove the following bound for multiplication over fields.

\begin{theorem} \label{thm:dnd_conditional_lowerbound_fields}
	Assume that we can solve \dnd{} over fields in
	$O(\frac{d^2}{n^{2/\omega+\epsilon}})$ rounds, for some $\epsilon>0$, in the low-bandwidth model with $n$ machines. Then we can also solve dense
	$n\times n$ matrix multiplication over fields in
	$O(n^{2-2/\omega-\epsilon/2})$ rounds.
\end{theorem}

\subsection{Unconditional Lower Bound for \texorpdfstring{\boldmath\dnd{}}{⟨d,n,d⟩}}
\label{subsec:dnd-lowerbound-unconditional}

Here we aim to prove the following.

\begin{theorem} \label{thm:dnd_unconditional_lowerbound}
	In the low-bandwidth model with $n$ machines, any algorithm that computes the product \dnd{} over fields or semirings requires \dndlowuncond{} communication rounds.
\end{theorem}

\begin{proof}
We can use the idea of the conditional lower bound, but use the unconditionally true bound from \Cref{cor:nnn_unconditional_lowerbound} that square matrix multiplication cannot be done in $o(n)$ rounds.
\end{proof}